\documentclass{article}

\usepackage{amsfonts}
\usepackage{amsmath}
\usepackage{amssymb}
\usepackage{graphicx}
\usepackage{verbatim}
\usepackage{float}
\usepackage{xcolor}

\usepackage{listings}
\usepackage{amsfonts}
\usepackage{amsmath}
\usepackage{amssymb}
\usepackage{enumerate}

\newcommand{\ie}{i.e.}

\newcommand{\myset}[2]{ \left\{ #1 \left| \, #2 \right. \right\} }

\newcommand{\prefix}{\sqsubseteq}

\newcommand{\ignore}[1]{}

\newcommand{\N}{\mathbb{N}}
\newcommand{\Q}{\mathbb{Q}}

\newcommand{\dimfs}{\mathrm{dim}_\mathrm{FS}}

\newcommand{\dimFS}{\dimfs}

\newcommand{\polylog}{{\mathrm{polylog}}}

\newcommand{\regSS}{S^\infty}
\newcommand{\expSS}{S^{\mathrm{exp}}}
\newcommand{\strexpSS}{S^{\mathrm{exp}}_{\mathrm{str}}}
\newcommand{\strSS}{S^\infty_{\mathrm{str}}}

\newcommand{\D}{{\mathcal{D}}}

\newcommand{\strings}{\{0,1\}^*}

\usepackage{ifthen}

\newtheorem{theorem}{Theorem}[section]

\newtheorem{property}[theorem]{Property}

\newtheorem{lemma}[theorem]{Lemma}
\newtheorem{corollary}[theorem]{Corollary}
\newtheorem{claim}[theorem]{Claim}

\newtheorem{construction}[theorem]{Construction}

\newenvironment{definition}
{ {\noindent {\bf Definition.}} } {  }

\newenvironment{example*}[1]
{ {\noindent {\bf Example #1.}} } {  }
\newenvironment{claim*}[1]
{ {\noindent {\bf Claim #1.}} } {  }
\newenvironment{theorem*}
{ {\noindent {\bf Theorem.}} } {  }
\newtheorem{observation}[theorem]{Observation}

\newenvironment{proof}[1][xyzzy]
{
{\noindent {\bf Proof}%
\ifthenelse{\equal{#1}{xyzzy}}{{\bf .}}{~(#1).}} } { \hfill $\Box$ }

\newenvironment{proofof}[1]
{
{\noindent {\bf Proof #1.}%
}} { \hfill $\Box$ }

\newcommand{\sigb}{\Sigma_b}

\newcommand{\upto}{\upharpoonright}

\DeclareMathOperator{\real}{real}

\DeclareMathOperator{\seq}{seq}

\newcommand{\dlz}{d_{\mathrm{LZ}(b)}}
\newcommand{\dlzp}{d_{\mathrm{LZ+}(b)}}
\DeclareMathOperator{\head}{head} \DeclareMathOperator{\tail}{tail}

\title{Computing Absolutely Normal Numbers in Nearly Linear Time }
\author{Jack H. Lutz\footnote{Department of Computer Science, Iowa State
University, Ames, IA 50011 USA. lutz@cs.iastate.edu.} \and Elvira
Mayordomo\footnote{Departamento de Inform\'atica e Ingenier\'ia de
Sistemas, Instituto de Investigaci\'on en Ingenier\'{\i}a de
Arag\'on, Universidad de Zaragoza, 50018 Zaragoza, SPAIN.
elvira@unizar.es.}}

\date{}

\begin{document}

\maketitle

\begin{abstract}
A real number $x$ is {\sl absolutely normal\/} if, for every base
$b\ge 2$, every two equally long strings of digits appear
           with equal asymptotic frequency in the base-$b$ expansion of $x$.
           This paper presents an explicit algorithm that
           generates the binary expansion of an absolutely normal number $x$,
           with the $n$th bit of $x$ appearing after
           $n\polylog(n)$ computation steps.
           This speed
           is achieved by
           simultaneously computing and diagonalizing against a martingale
 that incorporates Lempel-Ziv parsing algorithms in all bases.

\end{abstract}

{\bf Keywords:} algorithms, computational complexity, Lempel-Ziv
parsing, martingales, normal numbers

Declarations of interest: none

\section{Introduction}\label{sec1}

In 1909 Borel\footnote{Borel's original definition was that $x$ is
normal in base $b$ if $x$ is simply normal in all bases $b, b^2 b^3,
\ldots$ Our definition here is well known to be equivalent to
Borel's.}
 \cite{Bore09}\ defined a real number $\alpha$ to be
{\sl normal\/} in base $b$ ($b\ge 2$) if, for every $m\ge 1$ and
every length-$m$ sequence $w$ of base-$b$ digits, the asymptotic,
empirical frequency of $w$ in the base-$b$ expansion of $\alpha$ is
 $b^{-m}$. Borel defined $\alpha$ to be {\sl absolutely normal\/} if
it is normal in every base $b\ge 2$. (This clearly anticipated the
fact, proven a half-century later, that a real number may be normal
in one base but not in another \cite{Cass59, Schm60}.) Borel's proof
that {\sl almost every\/} real number (i.e., every real number
outside a set of Lebesgue measure 0) is absolutely normal was an
important milestone in the prehistory of Kolmogorov's development of
the rigorous, measure-theoretic foundations of probability theory
\cite{Kolmo50}. For example, it is section 1 of Billingsley's
influential textbook \cite{Bill95}. The recent book \cite{Buge12}\
provides a good exposition of the many aspects of current research
on normal numbers.

Borel's proof shows that absolutely normal numbers are commonplace,
i.e., that a ``randomly chosen'' real number is absolutely normal
with probability 1.   Rational numbers cannot be normal in even a
single base $b$, since their base-$b$ expansions are eventually
periodic, but computer analyses of the expansions of $\pi$, $e$,
$\sqrt{2}$, $\ln 2$, and other irrational numbers that arise in
common mathematical practice suggest that these numbers are
absolutely normal \cite{BorBai08}. Nevertheless, no such ``natural''
example of a real number has been proven to be normal in {\sl any\/}
base, let alone absolutely normal.  The conjectures that every
algebraic irrational is absolutely normal and that $\pi$ is
absolutely normal are especially well known open problems
\cite{BorBai08, Buge12, Wago85}.

The work reported here concerns a much newer problem, namely, the
{\sl complexity\/} of explicitly computing a real number that is
provably absolutely normal, even if it is not  natural in the above
informal sense. Sierpinski and Lebesgue gave explicit constructions
of absolutely normal numbers in 1917 \cite{Sier17, Lebe17}, but
these were intricate limiting processes that offered no complexity
analyses (coming two decades before the theory of computing) and
little insight into the nature of the numbers constructed. In a 1936
note that was not published in his lifetime, Turing \cite{Turi13}\
gave a constructive proof that almost all real numbers are
absolutely normal and then {\sl derived\/} constructions of
absolutely normal numbers from this proof. Moreover, although Turing
does not mention Turing machines or computability, the note is
typed, with equations handwritten by him, on the back of a draft of
his paper on computable real numbers \cite{Turi36}, so it is
reasonable to interpret ``constructively'' in a
computability-theoretic sense. And in fact his proof, with 2007
corrections by Becher, Figueira, and Picchi \cite{BeFiPi07},
explicitly computes an absolutely normal number $x$. However, this
algorithm is very inefficient, requiring a number of steps that is
doubly exponential in $n$ to compute the $n$th bit of $x$.
(Independently, Knuth \cite{knut65}\ published in 1965 an explicit
construction of an absolutely normal number, also inefficient).

Some 75 years  passed between Turing's algorithm and more efficient
ones. It was only in 2013 that Becher, Heiber, and Slaman
\cite{BeHeSl13}\ published an algorithm that computes an absolutely
normal number in polynomial time. Specifically, this algorithm
computes the binary expansion of an absolutely normal number $x$,
with the $n$th bit of $x$ appearing after $O(n^2f(n))$ steps for any
computable unbounded nondecreasing function $f$. (Unpublished
polynomial-time algorithms for computing absolutely normal numbers
were  announced independently by Mayordomo \cite{Mayo13}\ and
Figueira and Nies \cite{FiNi13, FigNie15}\ at about the same time.)
We omit here extensive work on the discrepancy, that is, the order
of convergence to normality of an absolutely normal number and its
tradeoff with the time complexity of the construction of the
corresponding number (see the latest results in \cite{ABSS17},
\cite{MaScTi18} and their references).

In this paper we present a new algorithm that provably computes an
absolutely normal in nearly linear time.  Our algorithm computes the
binary expansion of an absolutely normal number $x$, with the
 first to $n$th bits  of $x$ appearing after $O(n\polylog n)$ steps. The
term ``nearly linear time'' was introduced by Gurevich and Shelah
\cite{GurShe89}. In that paper they showed that, while linear time
computability is very model-dependent, nearly linear time is very
robust.  For example, they showed that random access machines,
Kolmogorov-Uspensky machines, Schoenhage machines, and random-access
Turing machines share exactly the same notion of nearly linear time.

The novelty of our algorithm is its use of the Lempel-Ziv parsing algorithm to achieve its nearly linear time bound.  For
each base $b\ge 2$, we use a martingale (betting strategy) that employs the Lempel-Ziv parsing algorithm and is implicit in
the work of Feder \cite{Fede91}. This base-$b$ Lempel-Ziv martingale succeeds exponentially when betting on the successive
digits of the base-$b$ expansion of any real number that is not normal in base $b$.  Our algorithm simultaneously computes
and diagonalizes against (limits the winnings of) a martingale that incorporates efficient proxies of all these
martingales, thereby efficiently computing a real number that is normal in every base.

The structure of this paper is based on the main result proof.
Building on the base-$b$ normality characterization in terms of
finite state martingales and the universality of Lempel-Ziv
martingale, we need to construct a conservative version of the
Lempel-Ziv martingale that does not fluctuate too much and then
establish a base change method for such conservative martingales.
Finally a careful efficient combination of all resulting strategies
for different bases is needed.

The rest of this paper is organized as follows.  Section \ref{sec2}\
presents the base-$b$ Lempel-Ziv martingales and their main
properties.  Section \ref{sec3}\  shows how to transform a base-$b$
Lempel-Ziv martingale into a base-$b$  supermartingale
 with an efficiently computable nondecreasing savings account
that is unbounded whenever the base-$b$ Lempel-Ziv martingale
succeeds exponentially. Section \ref{sec4}\ develops an efficient
method for converting a base-$b$ martingale with an efficiently
computable savings account to a base-2 martingale that succeeds
whenever the base $b$ savings account is unbounded. Section
\ref{sec5}\ presents an algorithm that exploits the uniformity of
these constructions to efficiently and simultaneously compute (a) a
single base-2 martingale $d$ that succeeds on the binary expansion
of every real number $x$ for which some base-$b$ martingale succeeds
on the base-$b$ expansion of $x$ and (b) a particular real number
$x$ on which binary expansion $d$ does not succeed. This $x$ is,
perforce, absolutely normal. Section \ref{sec6}\ presents an open
problem related to our work.

For our complexity arguments, we use a (log-cost)  RAM model. By
\cite{AngVal79}\ a nearly linear time bound on this model is
equivalent a nearly linear time bound for any of the robust models
in \cite{GurShe89}.

\section{Lempel-Ziv Martingales}\label{sec2}

For each base $b\ge 1$ we let $\Sigma_b=\{0, 1, \ldots, b-1\}$ be
the alphabet of {\sl base-$b$\/}  {\sl digits}. We write
$\Sigma_b^*$ for the set of all (finite) {\sl strings\/} over
$\Sigma_b$ and $\Sigma_b^{\infty}$ for the set of all (infinite)
{\sl sequences\/} over $\Sigma_b$. We write $|x|$ for the length of
a string or sequence $x$, and we write $\lambda$ for the {\sl empty
string}, the string of length 0. For $x\in
\Sigma_b^*\cup\Sigma_b^{\infty}$ and $0\le i\le j<|x|$, we write
$x[i .. j]$ for the string consisting of the $i$th through $j$th
digits in $x$. For $x\in \Sigma_b^*\cup\Sigma_b^{\infty}$ and $0\le
n<|x|$, we write $x\upto n=x[0..n-1]$. For $w\in\Sigma_b^*$ and
$x\in \Sigma_b^*\cup\Sigma_b^{\infty}$, we say that $w$ is  a {\sl
prefix\/} of $x$, and we write $w\sqsubseteq x$, if $x\upto |w|=w$.

Let $\D$ be the set of dyadic rationals. Let $f: \Sigma_b^* \to [0,
\infty)$ be a function. $f$ is {\sl nearly linear time computable}
if there exists $a,c> 1$ and $\widehat{f}: \Sigma_b^* \to\D$ such
that $\widehat{f}$ is exactly nearly linear time computable and for
all $w\in \Sigma_b^*$, $|\widehat{f}(w)-f(w)|\le a/{|w|^c}$.

Let $f: \Sigma_b^* \to [0, \infty)$ be a function. $f$ is {\sl
online nearly linear time computable} if there exists $g: \Sigma_b^*
\to [0, \infty)$  such that
\begin{enumerate}
\item $f, g$ are nearly linear time computable,
\item for $w\in \Sigma_b^*, a\in
\Sigma_b$, $f(wa)$ and $g(wa)$ can be computed in polylogarithmic
time from $w, f(w), g(w)$.

\end{enumerate}

A ({\sl base}-$b$) {\sl martingale\/} is a function $d:
\Sigma_b^*\to [0, \infty)$ satisfying
\begin{equation}\label{cond21}d(w)=\frac{1}{b}\sum_{a\in\Sigma_b}d(wa)\end{equation}
for all $w\in\Sigma_b^*$. (This is the original martingale notion
introduced by Ville \cite{Vill39}\ and implicit in earlier papers of
L\'evy \cite{Levy35, Levy37}. Its relationship to Doob's subsequent
modifications \cite{Doob40}, which are the ``martingales'' of
probability theory, is explained in \cite{HitLut06}\ along with the
reason why Ville's original notion is still essential for
algorithmic information theory.) Intuitively, a base-$b$ martingale
$d$ is a {\sl strategy for betting\/} on the subsequent digits in a
sequence $S\in\Sigma_b^{\infty}$, with the strategy encoded in such
a way that $d(S\upto n)$ is the amount of money that a gambler using
the strategy $d$ has after the first $n$ bets. The condition
(\ref{cond21}) says that the payoffs for these bets are {\sl fair\/}
in the sense that the conditional expectation of $d(wa)$, given that
$w$ has occurred (and assuming that the digits $a\in\Sigma_b$ are
equally likely), is $d(w)$.

A function $g: \Sigma_b^*\to [0, \infty)$ (which may or may not be a
martingale) {\sl succeeds\/} on a sequence $S\in\Sigma_b^{\infty}$
if \begin{equation}\label{equa22}\limsup_{n\to\infty}g(S\upto
n)=\infty,\end{equation}\ie, if its winnings on $S$ are unbounded.
The {\sl success set\/} of a function $g: \Sigma_b^* \to [0,
\infty)$ is \[\regSS[g]=\myset{S\in\Sigma_b^{\infty}}{g \mbox{
succeeds on }S}.\] A function $g: \Sigma_b^* \to [0, \infty)$ {\sl
succeeds exponentially\/} on a sequence $S\in \Sigma_b^{\infty}$ if
\begin{equation}\label{equa23}\limsup_{n\to \infty}\frac{\log
g(S\upto n)}{n}>0,\end{equation} \ie, if its winnings on $S$ grow at
some exponential rate, perhaps with recurrent setbacks. The {\sl
exponential success set\/} of a function $g: \Sigma_b^* \to [0,
\infty)$ is \[\expSS[g]=\myset{S\in\Sigma_b^{\infty}}{g \mbox{
succeeds exponentially on }S}.\]

The {\sl $f(n)$ success set\/} of a function $g: \Sigma_b^* \to [0,
\infty)$ is
\[S^{f(n)}[g]=\myset{S\in\Sigma_b^{\infty}}{\limsup_{n\to
\infty}\frac{\log g(S\upto n)}{\log f(n)}\ge 1}.\] Note that $\expSS[g]=\cup_{\epsilon>0}S^{2^{\epsilon n}}[g]$.

For technical reasons we will also need to consider the notion of
supermartingale, which in many contexts turns out to be equivalent
to the notion of martingale.

A ({\sl base}-$b$) {\sl supermartingale\/} is a function $d:
\Sigma_b^*\to [0, \infty)$ satisfying
\begin{equation}\label{cond21}d(w)\le\frac{1}{b}\sum_{a\in\Sigma_b}d(wa)\end{equation}
for all $w\in\Sigma_b^*$.

\begin{lemma}\label{lem2x} For each online nearly linear time computable supermartingale $d$ there is an online nearly linear time computable martingale
$d'$ such that for every $w\in \Sigma_b^*$ $d'(w)\ge d(w)$. If for
every $w\in \Sigma_b^*$ $d(w)\le a |w|^c$, then $d'(w)\le a
|w|^{c+1}$. If for some $m, C$, $d(w)=C$ for $|w|\le m$ then
$d'(w)=C$ for $|w|\le m$.
\end{lemma}

\begin{proofof}{(proof sketch)}

We define $d'$ recursively as follows,
 $d'(\lambda)=d(\lambda)$ and for $w\in \Sigma_b^*, a\in
\Sigma_b$

\[d'(wa) = d'(w) + d(wa)  -1/b \sum_{b'}{d}(wb')\]

$d'$ is online nearly linear computable. It can be proven by
induction
 that $d'$ is a martingale and that for all for $w\in \Sigma_b^*$,  $d'(w)\ge
 d(w)$ and $d'(w)\le \sum_{u\sqsubseteq w} d(u)$.

\end{proofof}

 A function $g: \Sigma_b^*\to [0,
\infty)$ {\sl succeeds strongly\/} on a sequence $S\in\Sigma_b^{\infty}$ if (\ref{equa22}) holds with the limit superior
replaced by a limit inferior \ie, if the winnings converge to $\infty$. A function $g: \Sigma_b^*\to [0, \infty)$ {\sl
succeeds strongly exponentially\/} on a sequence $S\in\Sigma_b^{\infty}$ if (\ref{equa23}) holds with the limit superior
replaced by a limit inferior \ie, if the winnings grow at exponential rate. The {\sl strong success sets\/} $\strSS[g]$
 and $\strexpSS[g]$ of a function $g: \Sigma_b^*\to [0, \infty)$ are defined in the now-obvious manner. It is clear that the
inclusions
\[\strexpSS[g]\subseteq \expSS[g]\subseteq \regSS[g]\] and
\[\strexpSS[g]\subseteq \strSS[g]\subseteq \regSS[g]\]
hold for all $g: \Sigma_b^*\to [0, \infty)$.

For each base $b\ge 2$ the {\sl base-$b$ Lempel-Ziv martingale\/} is a particular martingale $\dlz$ based on the Lempel-Ziv
parsing algorithm \cite{LemZiv78}, as we now explain.

Formally, $\dlz$ is computed by the algorithm in Figure \ref{figu1},
but some explanation is appropriate here. The algorithm is written
with several instances of parallel assignment. For example, the
second line initializes $x$, $L(x)$, and $d$ to the values
$\lambda$, 1, and 1, respectively. The items $T$, $j$, and $x(j)$
are not needed for the computation of $\dlz(w)$, but they are useful
for understanding and analyzing the algorithm.

\begin{figure}[H]
\begin{lstlisting}[mathescape=true]
         input $w\in\Sigma_b^*$;
         $x, L(x), d = \lambda, 1, 1$;
         $T, j = \{\lambda\}, 0$;
         while true do
         begin
            if $w= \lambda$ then output $d$ and halt;
            if $L(x)=1$ then
            begin
                $L(x)=b$;
                for each $0\le a<b$ do $L(xa)=1$;
                $T, x(j), j = T\cup \{x\}\Sigma_b, x, j+1$;
                $x=\lambda$;
            end
            else
            begin
                $a,w = \head(w), \tail(w)$;
                $L(x), x, d = L(x)+b-1, xa, \frac{b L(xa)}{L(x)}d$
            end
        end
      \end{lstlisting}\caption{Algorithm for computing $\dlz(w)$.}\label{figu1}
      \end{figure}

The growing set $T$ of strings in $\Sigma_b^*$ always contains all
the prefixes of all its elements, so it is a tree. We envision this
tree as being oriented with its root at the top and the immediate
children $v0, v1, \ldots, v(b-1)$ of each interior vertex $v$ of $T$
displayed left-to-right below $T$. The {\sl dictionary\/} of the
algorithm is the current set of leaves of $T$.

The string $x$ in the algorithm is always an element of (\ie,
location in) the tree $T$, and $L(x)$ is always the number of leaves
of $T$ that are descendants of $x$. We regard $x$ as a descendant of
itself, so $x$ is a leaf if and only if $L(x)=1$.

It is clear that $\dlz(\lambda)=1$. In fact, the algorithm's
successive values of $d$ are the values $\dlz(u)$ for successive
prefixes $u$ of the input string $w$. More precisely, if $w_t$ and
$d_t$ are the values of $w$ and $d$ after $t$ executions of the
else-block, then $w=(w\upto t)w_t$ and $d_t=\dlz(w\upto t)$.

For $w\in\Sigma_b^*$ we define the tree $T(w)$ as follows. If
$w=\lambda$, then $T(w)=\{\lambda\}$. If $w=w'a$, where
$w'\in\sigb^*$ and $a\in\sigb$, then  $T(w'a)$
 is the value of $T$ when the algorithm terminates on input
$w'$. (Note that this is one step before it terminates on input
$w'a$.) For $w\in\sigb^*$ we define $D(w)$ to be the number of
leaves in $T(w)$.  For each $x$ in $T(w)$,  $L(x,w)$ is the number
of leaves of $T(w)$ that are descendants of $x$.

The computation is divided into ``epochs''. At the beginning of each epoch, the string $x$ is $\lambda$, \ie, it is located
at the root of $T$. The string $x$ then takes successive digits from whatever is left of $w$ (because $a, w = \head(w),
\tail(w)$ removes the first digit of $w$ and stores it in $a$), following this path down the tree and updating $d$ at each
step, until $w$ is empty (the end of the last epoch) or $x$ is a leaf of $T$. In the latter case, the $j$th epoch is over,
the $b$ children $x0, x1, \ldots, x(b-1)$ are added to $T$ as new leaves, $x$ is the $j$th {\sl phrase} $x(j)$ of $w$, and
$x$ is reset to the root $\lambda$ of $T$.

When the algorithm terminates, it is clear that exactly one of the
following things must hold.

\begin{description}
\item[(a)] $w=\lambda$.
\item[(b)] $w= x(0)\ldots x(j-1)$.
\item[(c)] $w= x(0)\ldots x(j-1)u$ for some nonempty interior vertex
$u$ of $T(w)$.
\end{description}
In case (a) or (b) we call $w$ a {\sl full parse}. In case (b) or
(c) we call $x(0), \ldots, x(j-1)$ the {\sl full phrases\/} of $w$.
In case (c) we call $u$ the {\sl partial phrase \/} of $w$.

Notice that the function $h(w)= (T(w), j(w), D(w))$ is online nearly
linear time computable.  Notice that our algorithm
 does not work in nearly linear time  when
computing $d$ as a product of $|w|$ factors. We give below an
alternative characterization of $\dlz$ that  will be useful later.

Define the set $A_b=\myset{1+(b-1)r}{r\in\N}$ and the generalized
factorial function $fact_b: A_b\to A_b$ by
\[fact_b(1+(b-1)r)=\prod_{k=1}^r(1+(b-1)k)\] for all $r\in\N$.

\begin{observation}\label{obser21} For all $n\in A_b$,
 \begin{equation}\label{equa25}1\le \frac{fact_b(n)}{e^{\frac{1}{b-1}}\left(\frac{n}{e}\right)^{\frac{n}{b-1}}}\le
n.\end{equation}
\end{observation}

Using Euler-Maclaurin formula we also have
\begin{observation}\label{obser21p} For all $n\in A_b$,
 \begin{align}\label{equa25p}&fact_b(n)= \\ &C \cdot n^{\frac{n}{b-1}} \cdot
 e^{-{(n-1)}{(b-1)}} \cdot n^{1/2}\cdot e^{(b-1)/(12n)} \cdot e^{-(b-1)^3/(720n^3)}\cdot e^{(b-1)^5/(30240n^5)}
\cdot e^{O(1/n^7)}.\end{align}
\end{observation}

Using the terms in Observation \ref{obser21p}\ we define
\begin{equation}\label{equa25p2}\widehat{fact_b}(n)= C \cdot n^{\frac{n}{b-1}} \cdot
 e^{-{(n-1)}{(b-1)}} \cdot n^{1/2}\cdot e^{(b-1)/(12n)} \cdot e^{-(b-1)^3/(720n^3)}\cdot e^{(b-1)^5/(30240n^5)}.\end{equation}
All terms in the definition of $\widehat{fact_b}(n)$ are computed
with an approximation of $e^{O(1/n^7)}$, that is, all terms in the
exponents will have precision $7\log n +O(1)$.

\begin{lemma}\label{lemm21}(Feder \cite{Fede91}) Let $w\in\sigb^*$.\begin{enumerate}
\item If $w$ is a full parse, then
\[\dlz(w)=\frac{b^{|w|}}{fact_b(D(w))}.\]
\item If $w$ is not a full parse and $u$ is its partial phrase, then
\[\dlz(w)=\frac{b^{|w|}}{fact_b(D(w))}L(u,w),\]
where  $L(u,w)$  is the number of leaves below $u$ in
$T(w)$.\end{enumerate}
\end{lemma}

The following lemma follows from Lemma \ref{lemm21}.

\begin{lemma}\label{lemm22b} For $S\in \sigb^{\infty}$ and $\alpha\in(0,1)$ the following
three conditions are equivalent.
\begin{description}
\item[(a)] $S\in S^{b^{(1-\alpha)n}}[\dlz]$.
\item[(b)] There exist infinitely many full parses
$w\sqsubseteq S$ for which \[D(w)\log_b |w|<\alpha(b-1)|w|.\]
\item[(c)] There exist infinitely many full parses
$w\sqsubseteq S$ for which \[D(w)\log_b D(w)<\alpha(b-1)|w|.\]
\end{description}\end{lemma}

\begin{corollary}\label{lemm22} For $S\in \sigb^{\infty}$ the following
three conditions are equivalent.
\begin{description}
\item[(a)] $S\in\expSS[\dlz]$.
\item[(b)] There exist $\alpha<1$ and infinitely many full parses
$w\sqsubseteq S$ for which \[D(w)\log_b |w|<\alpha(b-1)|w|.\]
\item[(c)] There exist $\alpha<1$ and infinitely many full parses
$w\sqsubseteq S$ for which \[D(w)\log_b D(w)<\alpha(b-1)|w|.\]
\end{description}\end{corollary}


We conclude this section by explaining the connection between the
Lempel-Ziv martingales and normality. First, Schnorr and Stimm
\cite{SchSti72}\ defined (implicitly) a notion of {\sl finite-state
base-$b$ martingale\/} and proved that every sequence
$S\in\sigb^{\infty}$ obeys the following dichotomy.
\begin{enumerate}
\item If $S$ is normal, then no finite-state base-$b$ martingale
succeeds on $S$. (In fact, every finite-state base-$b$ martingale
decays exponentially on $S$.)
\item If $S$ is not normal, then some finite-state base-$b$ martingale
succeeds exponentially on $S$.
\end{enumerate}

Some twenty years later, Feder \cite{Fede91}\ defined (implicitly)
the Lempel-Ziv martingale $\dlz$ and proved (implicitly) that $\dlz$
is at least as successful {\sl on every sequence\/} as every
finite-state base-$b$ martingale. That is, for finite-state base-$b$
martingale $d$, the inclusions
\[\regSS[d]\subseteq \regSS[\dlz], \ \strSS[d]\subseteq \strSS[\dlz],\]
\[\expSS[d]\subseteq \expSS[\dlz],\ \strexpSS[d]\subseteq
\strexpSS[\dlz]\]   all hold. This, together with Schnorr and
Stimm's dichotomy result, implies that $\dlz$ succeeds exponentially
on every non-normal sequence in $\sigb^{\infty}$. Hence a real
number $x$ is absolutely normal if none of the martingales $\dlz$
succeed exponentially on the base-$b$ expansion of $x$.

In order to avoid time bounds that are dependent on the alphabet
size $b$, we will consider the following variant of $\dlz$,

\[
\dlzp(w) = \begin{cases}  1, & \quad\text{ if }|w|\le 2^b\\
\frac{\dlz(w)}{\dlz(w\upto 2^b)},   & \quad \text{ if } |w|> 2^b.
\end{cases}
\]

Notice that $\regSS[\dlz]=\regSS[\dlzp]$ and that for  any a.e.
unbounded $f$, $S^{f(n)}[\dlz]=S^{f(n)}[\dlzp]$.

Notice that if $\log(|w|)> b$ then a polynomial bound on $b$ is a
polylogarithmic bound on $|w|$.

\section{Savings Accounts}\label{sec3}

In this section we construct a conservative version of the the
Lempel-Ziv martingale  $\dlzp$  consisting of a new
 supermartingale  $d'$  that can be
smaller than  $\dlzp$  but that has a savings account in the  sense
explained next. We will need this conservative version in the base
change transformation in section \ref{sec4}.

\begin{definition}
A  function $g : \sigb^* \to [0,\infty)$ is a {\sl savings
account\/} of a  supermartingale  $d : \sigb^* \to [0,\infty)$ if
$g$ is nondecreasing with respect to substring order and, for every
$w\in\sigb^*$, $d(w)\ge g(w)$.
\end{definition}

In the following construction we use Observation \ref{obser21}\
 and Lemma \ref{lemm21}\  to get a far more conservative version of  $\dlzp$.  We define a
function $goal(w)$ such that \[b\le \dlzp(w)b^{-goal(w)}\le
b^6|w|^2\] and then a nondecreasing upper bound $taken(w)\ge
goal(w)$ such that for every $S$, $taken(S\upto n)$ coincides with
$goal(S\upto n)$ infinitely often.

\begin{construction}\label{cons31}
Let  $d=\dlzp$  be the base-$b$-Lempel-Ziv martingale. We define a
new  supermartingale
 $d'=e'+g'$ as follows.

We first define $e'$.  Let $w\in\sigb^*$. For $|w|\le 2^b$,
$e'(w)=b$,  $taken(w)=0$.  For $|w|> 2^b$,
 let $w=x(0) \ldots x(j-1)u$, for $z=x(0) \ldots x(j-1)$  a
full parse and $u$ the
 partial or  the last
 full phrase of $w$. Let


 \begin{eqnarray*}goal(w)&= &|w|-\lceil D(w)(\log_b(D(w)))/(b-1)\rceil+\lfloor
 D(w)(\log_b(e))/(b-1)\rfloor
 \\&&-\lceil\log_b(D(w))\rceil-\lceil \log_b e^{\frac{1}{b-1}}\rceil
 -\lceil \log_b \dlz(w\upto 2^b)\rceil
 -1,\end{eqnarray*}

 \[taken(w)=\max\{taken(z), goal(w)\},\]%




%
\[e'(w)= {d(w)}b^{-taken(w)}.\]
%



Let $g'$ be defined as follows.
 Let $w\in\sigb^*, a\in\sigb$. For $|w|\le
2^b$, $g'(w)=0$ and for $|w|\ge 2^b$,
\[g'(wa)=\left\{\begin{array}{ll}g'(w)&\mbox{if }goal(wa)\le taken(w)
\\g'(w)+e'(w)\frac{b-1}{b}&\mbox{if
}goal(wa)> taken(w)
\end{array}\right.\]

\end{construction}

\begin{theorem}\label{theo32} Let $d'$ and $g'$ be as defined in Construction \ref{cons31}.
Then  $d'$ is a  supermartingale  and $g'$ is its savings account,
\[\expSS[\dlz]\subseteq \regSS[g'],\]  $d'$ 
is computable in
 an online  nearly  linear time bound that does not depend on
$b$, and there exists $a, c>1$ not depending on $b$ such that for
every $w\in\sigb^+$, $d'(w)\le
 b\cdot a\cdot |w|^c$, $d'(\lambda)=b$.

\end{theorem}

\begin{proofof}{of Theorem \ref{theo32}}
\begin{claim}
$d '=e'+g'$ is a  supermartingale  and $g'$ is a savings account for
$d'$.
\end{claim}

\begin{proof}
Let us prove that $d'$ is a  supermartingale. Note that by
definition of $D$, $goal(wa)$ does not depend on $a$. When
$goal(wa)\le taken(w)$ we have that
\begin{equation*}
\sum_{a\in\Sigma_b}e'(wa) =  \frac{e'(w)}{d(w)}
\sum_{a\in\Sigma_b}d(wa)= b\cdot \frac{e'(w)}{d(w)}\cdot d(w)=
b\cdot e'(w).\end{equation*}  Since $g'(wa)$ is constant the
martingale equality holds in this case.

In the second case,  when $goal(wa)> taken(w)$, since they are
integer values $goal(wa)\ge  taken(w)+1$.  We
 have that
\begin{eqnarray*}\sum_{a\in\sigb}d'(wa)&=\sum_{a\in\sigb} e'(wa)+
\sum_{a\in\sigb}
g'(wa)\\&=\frac{e'(w)}{d(w)}\sum_{a\in\sigb}b^{-goal(wa)+taken(w)}d(wa)+
\sum_{a\in\sigb} g'(wa)\\
&\le  e'(w) +b(g'(w)+e'(w)\frac{b-1}{b})= b(e'(w)+g'(w))= b
d'(w).\end{eqnarray*}

Since $e'$ is nonnegative, by definition  $g'$ is a nondecreasing
function. Therefore $g'$ is a savings account of $d'$.
\end{proof}

\begin{claim}\label{claimbounds}   For every
$w\in\sigb^*$ with $|w|>2^b$
\[e'(w)=d(w)b^{-taken(w)}\le b^6\cdot D(w)\cdot L(u,w).\]
\[d(w) b^{-goal(w)}\ge b\] for $u$ the partial or the
last full phrase of $w$.\end{claim}
\begin{proof}
Use that $taken(w)\ge goal(w)$,  Lemma \ref{lemm21}, and Observation
\ref{obser21}.
\end{proof}

\begin{claim}\label{claimun}
If $y\in\sigb^{\infty}$ and $goal(y\upto n)$ is unbounded then $y\in
\regSS[g']$.
\end{claim}

\begin{proof}

If $goal(y\upto n)$ is unbounded then infinitely often we use the
second case in the  definition of $g'$ and have that $taken(y\upto
n)=goal(y\upto n)> taken(y\upto (n-1))$, $e'(y\upto n)= d(y\upto n)
b^{-goal(y\upto n)}$, and $g'(y \upto n)= g'(y \upto (n-1)) +e'(y
\upto (n-1))\frac{b-1}{b}$.

Since $goal(y \upto (n-1))\ge goal(y \upto n)-1$, then $taken(y\upto
(n-1))=goal(y\upto (n-1))$. By Claim \ref{claimbounds}, $e'(y \upto
(n-1))\ge b$, therefore $g'(y \upto n)\ge g'(y \upto (n-1)) +{b-1}$.

Since $g'$ is monotonic, $y\in \regSS[g']$.

\end{proof}

\begin{claim} For every $\alpha\in (0,1)$,
$S^{b^{(1-\alpha)n}}[\dlz]\subseteq \regSS[g']$.
\end{claim}

\begin{proof}
If $y\in S^{b^{(1-\alpha)n}}[\dlz]$ then by Lemma \ref{lemm22b}\
for infinitely many $n$,  $D(y\upto n)\log_b(D(y\upto
n))<\alpha(b-1) n$.

Notice that therefore $goal(y\upto n)$ is unbounded and by Claim
\ref{claimun}\ $y\in \regSS[g']$.

\end{proof}

\begin{claim}
  For every $w\in\sigb^*$, $e'(w)\le
 e'(\lambda)\log(|w|)^6|w|^2$ and
   $g'(w)\le
\sum_{v\sqsubseteq w}e'(v)$. Therefore $d'(w)$ is polinomially
bounded.

\end{claim}

\begin{claim}$d'$ can  be computed in   online
nearly linear time (in time $n\log^c n$ for $c$ not depending on
$b$).
\end{claim}
\begin{proof}

We first proof that $e': \sigb^* \to [0, \infty)$  can be computed
in online nearly linear time. Let

\[\widehat{e'}(w)=\frac{b^{|w|}}{\widehat{fact}_b(D(w))}L(u,w)b^{-taken(w)}.\]

Then

\begin{eqnarray*}|e'(w)-\widehat{e'}(w)|&\le& b^{|w|-taken(w)}L(u,w) \left|
\frac{1}{{fact}_b(D(w))}-\frac{1}{\widehat{fact}_b(D(w))}
\right|\\
&\le& b^{|w|-taken(w)}L(u,w) \left|
\frac{\widehat{fact}_b(D(w))-{fact}_b(D(w))}{{fact}_b(D(w))\cdot\widehat{fact}_b(D(w))}
\right| \\
&\le& b^{|w|-taken(w)}L(u,w) \left|
\frac{1-e^{O(1/D(w)^7)}}{{fact}_b(D(w))} \right| =e'(w)
|1-e^{O(1/D(w)^7)}|\\
&\le&e'(w) \frac{1}{CD(w)^7}+O(1/D(w)^{14})\\
&\le& \frac{1}{CD(w)^4}+O(1/D(w)^{11})
\end{eqnarray*}

Notice that $D(w)\ge \sqrt{|w|}$, therefore,

\[|e'(w)-\widehat{e'}(w)|\le b/C 1/|w|^{2}+O(1/|w|^5).\]

Since $\widehat{e'}$ is computable in online nearly linear time,
$e'$ is too.

Again, for $w\in\sigb^*$, let  $w=x(0) \ldots x(j-1)u$, for $z=x(0)
\ldots x(j-1)$   a full parse and $u$ a partial or full phrase of
$w$.
 If
$goal(w)>taken(z)$, let $t$ be  the shortest
 such that $t\sqsubseteq u$ and $goal(zt)> taken(z)$.


Notice  that  by Lemma \ref{lemm21}\ and the definition of $g'$,

\[g'(w)=\left\{\begin{array}{ll}g'(z)&\mbox{if }goal(w)\le taken(z)\\g'(z)+\sum_{v\sqsubseteq u, |zv|\ge |zt|}e'(zt)/L(t,w)\frac{b-1}{b}L(v,w)&\mbox{if
}goal(w)>taken(z),
\end{array}\right.\]
where $L(v,w)$  is the number of leaves below $v$ in $T(w)$. Given
precomputed values for $f(u)=\sum_{v\sqsubseteq u}L(v,w)$,
 the value of $g'(w)$ can be easily computed in
 online
 nearly linear time.

\end{proof}

This completes the proof of Theorem \ref{theo32}.
\end{proofof}

\section{Base change}\label{sec4}

We  use infinite sequences over $\Sigma_b$ to represent real numbers
in  $[0,1)$.  For this, we associate each string $w \in \Sigma_b^*$
with the half-open interval $[w]_b$ defined by
 $[w]_b = [x , x+ b^{-|w|})$,  for
$x=\sum\limits_{i=1}^{|w|} w[i-1] b^{-i}$. Each real number
 $\alpha
  \in  [0,1) $  is then represented by the unique sequence $\seq_b(\alpha)
  \in \Sigma_b^{\infty}$ satisfying $$w \prefix \seq_b(\alpha) \iff \alpha
  \in [w]_b$$ for all $w \in \Sigma_b^*$.  We have $$\alpha = \sum_{i=1}^\infty
  \seq_b(\alpha)[i-1] b^{-i}$$
and the mapping  $\seq_b:  [0,1) \to \Sigma_b^{\infty}$
  is a bijection. (Notice that $[w]_b$ being half-open
prevents double representations.)  We define $\real_b:
\Sigma_b^\infty \to [0,1)$ to be the inverse of $\seq_b$.
   A set of real numbers   $A \subseteq
  [0,1) $  is represented by the set $$\seq_b(A) = \{\seq_b(\alpha) \mid
  \alpha \in A\}$$ of sequences. If $X\subseteq\Sigma_b^\infty$ then
  \[\real_b(X)=\{\real_b(x) \mid
  x \in X\}.\]


\begin{construction}\label{cons41}
Let $d : \sigb^* \to [0,\infty)$ be a polynomially-bounded
martingale with a savings account $g$.

 We define $\gamma: \sigb^* \to [0,1]$ a probability
measure on $\sigb^{\infty}$,  $\gamma(w):= b^{-|w|}d(w)/d(\lambda)$.

Using the Carath\'eodory extension to Borel sets, $\gamma$ can be
extended to any interval $[a,c]$; we denote with $\widehat{\gamma}$
this extension. (In fact if we consider all $U\subseteq \sigb^*$
such that all $u,v\in U, u\ne v$ are incomparable and
$[u]_b\subseteq [a,c]$ for all $u\in U$, then
$\widehat{\gamma}([a,c])=\sup_{U}\sum_{u\in U}\gamma(u)$).

We define $\mu:\strings\to[0,1]$ by
$\mu(y)=\widehat{\gamma}([y]_2)$.

Finally we define $d^{(2)}:\strings \to [0, \infty)$ by $d^{(2)}(y)=
2^{|y|}\mu(y)$.

\end{construction}

\begin{theorem}\label{theo42} Let $d$ be a base-$b$ martingale that is polynomially
bounded   such that   $d(w)$ is  constant for $|w|\le 2^b$,
 and let $g$ be a savings account of $d$. Let $d^{(2)}$ be
defined from $d$ and $g$ as in Construction \ref{cons41}. Then

\[real_b(\regSS[g])-\Q\subseteq real_2(\regSS[d^{(2)}]).\]

Moreover, if $d$ 
is computable in 
  an online
 nearly linear time bound not depending on $b$, then so is
$d^{(2)}$.  If for all $w\in\sigb^+$, $d(w)\le a |w|^c$ then for all
$y\in\{0,1\}^+$, $d^{(2)}(y)\le 3 a |y|^c/d(\lambda)$.
 \end{theorem}

\begin{proofof}{of Theorem \ref{theo42}}

We will first show that Carath\'eodory extension of $d$ works for
sequences base change, and then approximate the resulting $d^{(2)}$
using $d$ restricted to $\sigb^m$ for a fixed $m$.

\begin{property}\label{prop30} If for all $w\in\sigb^*$, $d(w)\le a |w|^c$ then
for all $y\in\strings$, $d^{(2)}(y)\le 3 a |y|^c/d(\lambda)$.
\end{property}

Let $y\in \strings$. Let $A_y=\myset{w\in\sigb^*}{|w|=|y| \mbox{ and
}[w]_b\cap [y]_2\ne \emptyset}$. Then
\begin{eqnarray*}d^{(2)}(y)&\le& 2^{|y|}\sum_{w\in A_y}
b^{-|w|}d(w)/d(\lambda)\\&\le& 2^{|y|}\sum_{w\in A_y}
b^{-|w|}a|w|^c/d(\lambda) \\ &\le&
2^{|y|}(2^{-|y|}+2b^{-|y|})a|y|^c/d(\lambda)\\ &\le& 3
a|y|^c/d(\lambda)\end{eqnarray*}

\begin{property}\label{prop31} Let  $\alpha
  \in  [0,1)-\Q$.    If $\seq_b(\alpha)\in \regSS[g]$, then $\seq_2(\alpha)\in
  \regSS[d^{(2)}]$. \end{property}

  \begin{proof}
  Let $x=\seq_b(\alpha)\in \regSS[g]$. Let $y=\seq_2(\alpha)$.

We use here that $d$ has a savings account $g$, so if $g(x\upto
n)>m$ then for all $w$ with $x\upto n\sqsubseteq w$, $g(w)>m$.

 Let
$m\in\N$ and choose $n$ such that $g(x\upto n)>m$. Let $q$ be such
that $[y\upto q]_2\subseteq [x\upto n]_b$. Let us see that
$d^{(2)}(y\upto q)>m/d(\lambda)$.

Let $r\in\N$. Let $A_r^q=\myset{w\in\sigb^*}{|w|=r \mbox{ and
}[w]_b\subseteq [y\upto q]_2}$. Then

\begin{eqnarray*}d^{(2)}(y\upto q)&= &2^{q}\widehat{\gamma}([y\upto q]_2) = 2^{q}
\lim_r \sum_{w\in A_r^q} d(w)/d(\lambda)b^{-|w|}\\&\ge& 2^{q}
m/d(\lambda) \lim_r \sum_{w\in A_r^q} b^{-|w|}= 2^{q} m/d(\lambda)
2^{-q}= m/d(\lambda).\end{eqnarray*}

The last chain of equations holds because $[y\upto q]_2\subseteq
[x\upto n]_b$ and for every $w\in A_r^q$, $[w]_b\subseteq [y\upto
q]_2$, so $x\upto n\sqsubseteq w$ for any $w\in A_r^q$.

\end{proof}

We next compute $d^{(2)}$. For each $m\in\N$ we define
$\mu_m:\strings\to[0,1]$ by
\[\mu_m(y)=\sum_{|w|=m, [w]_b\cap [y]_2\ne \emptyset}\gamma(w).\]

\begin{claim}
For every $y\in \strings$ and $m\in\N$, $|\mu(y)-\mu_m(y)|\le 2
b^{-m}a m^c/d(\lambda)$.
\end{claim}

\begin{proof}
Let $a, c$ be such that $d(w)\le a |w|^c$ for every $w$ (using that
$d$ is polynomially bounded). Then since at most two strings $w$
with $|w|=m$ have the property that $[w]_b\cap [y]_2\ne \emptyset$
and $[w]_b\not\subseteq [y]_2$, we have
\[|\mu(y)-\mu_m(y)|\le 2 b^{-m}a m^c/d(\lambda).\]

\end{proof}

For each $m\in\N$ we define  $d^{(2)}_m:\strings\to[0, \infty)$ by
$d^{(2)}_m(y)=2^{|y|}\mu_m(y)$.

\begin{claim}
For some $c>0$, for every $y\in \strings$, for every $m\in\N$,
\[|d^{(2)}(y)-d^{(2)}_m(y)|\le 2^{|y|} 2^{-m\log b+c\log
m+1}/d(\lambda).\]
\end{claim}

\begin{corollary}\label{caprox} For some $c'>0$, for every $y\in \strings$,
\[|d^{(2)}(y)-d^{(2)}_{|y|/\log b + c'\log|y|}(y)|\le 1/|y|^3\cdot 1/d(\lambda).\]
\end{corollary}

\begin{proof}
Take $c'=(4+c)/\log b$ and use the previous claim.
\end{proof}

\begin{property}
For $m\in\N$,  $y\in \strings$, $d^{(2)}_m(y)$ can be computed by
considering a maximum of $2b$ neighbor strings $w\in\sigb^r$ for $r=
|y|/\log b$ to $m$, computing $d(w)$ for each of them and doing an
addition and a multiplication for each.
\end{property}

\begin{proof}
Consider $P$, the smallest prefix free set of strings $w\in\sigb^*$
such that $[w]_b\subseteq [y]_2$, and notice that $|w|\ge  |y|/\log
b$ for each such string. For each $r$ there are at most $2b-2$
strings of length $r$ in $P$ (otherwise we can replace some of them
by a single string of length $r-1$). For length $m$ we may need two
more strings $|w|=m, [w]_b\cap [y]_2\ne \emptyset$.

\end{proof}

\begin{corollary} For  $y\in \strings$,
$d^{(2)}_{|y|/\log b + c'\log|y|}(y)$ can be computed by considering
a maximum of $2b$ neighbor strings $w\in\sigb^r$ for $r= |y|/\log b$
to $|y|/\log b + c'\log|y|$, computing $d(w)$ for each of them and
doing an addition and a multiplication for each.
\end{corollary}

By Corollary \ref{caprox}\ $f(y)=d^{(2)}_{|y|/\log b +
c'\log|y|}(y)$ approximates $d^{(2)}(y)$ within a $1/|y|^3\cdot
1/d(\lambda)$ bound, and by the last corollary $f$ can be computed
in  online
 nearly linear time.
  Using
that $d(w)$ is constant for $|w|\le 2^b$ we have a  nearly linear
time bound independent of $b$.

 This concludes the proof of Theorem
\ref{theo42}.
\end{proofof}

\section{Absolutely Normal Numbers}\label{sec5}

In this section we give an algorithm that  diagonalizes against the Lempel-Ziv martingales for all bases  in nearly linear
time.

We use the following theorem  which is a union lemma for online
nearly linear martingales that works for a set of martingales that
is uniformly computable, uniformly approximated and uniformly
polinomially bounded.

\begin{theorem}\label{theo51}
Let $(d_k)_{k\in\N}$ be a sequence of base-2 martingales such that
for each of them there exists a function $\widehat{d_k}: \strings
\to [0,\infty)$ with the following  properties

\begin{enumerate}
\item $\widehat{d_k}$ is computable in
 an online
nearly linear time bound that does not depend on $k$.
\item   There is $a, c>1$ such that for  every $k\in\N$, $y\in\strings$ \[|d_k(y)-\widehat{d_k}(y)|\le
\frac{a}{|y|^c}.\]
\item $d_k(\lambda)=\widehat{d_k}(\lambda)= 1$ and  there is $a, c>1$ such that for  every $k\in\N$,
$y\in\strings$, $d_k(y)\le a |y|^c$.
\end{enumerate}

Then we can compute in  online
 nearly linear time a binary sequence $x$ such that, for every $k$, $x\not\in\regSS[d_k]$.
\end{theorem}

\begin{proof}

 Let $d: \strings \to [0,\infty)$ be defined
by
\[d(w)=\sum_{k=1}^{\infty} 2^{-k}{d_k}(w).\]
Then $d$ is a martingale. Let $\widehat{d}: \strings \to [0,\infty)$
be
\[\widehat{d}(w)=\sum_{k=1}^{(c+2)\log|w|+\log a} 2^{-k}\widehat{d_k}(w).\]
Notice that $\widehat{d}$ is computable in
 online
nearly linear time.

\begin{claim} For each $w\in\strings$, $|d(w)-\widehat{d}(w)|\le (a+1)/|w|^{c+1}$.\end{claim}

\begin{eqnarray*}|d(w)-\widehat{d}(w)|&\le&
\sum_{k=1}^{(c+2)\log|w|+\log a} 2^{-k}|{d_k}(w)-\widehat{d_k}(w)|+ \sum_{k=(c+2)\log|w|+\log a+1}^{\infty} 2^{-k}{d_k}(w)  \\
&\le& \sum_{k=1}^{(c+2)\log|w|+\log a} 2^{-k}a/|w|^c+ \sum_{k=(c+2)\log|w|+\log a+1}^{\infty} 2^{-k}a|w|^c   \\
&\le& a/|w|^c+ a|w|^c/(a|w|^{c+2}) =a/|w|^c+ 1/|w|^2.
\end{eqnarray*}

Our algorithm will diagonalize against $d$, constructing a binary
sequence $x$ as follows. If $x\upto n$ has been defined then choose
the next bit of $x$ as  $i\in\{0, 1\}$  that minimizes
$\widehat{d}((x\upto n)i)$.

\begin{claim} If $x\not\in\regSS[d]$ then for every $k$,
$x\not\in\regSS[d_k]$.\end{claim}

Let $n\in\N, k\in\N$, $d_{k}(x\upto n) \le 2^{k} d(x\upto n)$.

\begin{claim} $x\not\in\regSS[d]$.\end{claim}

Let $w\in\strings$. We prove that for $i\in\{0, 1\}$ chosen to
minimize $\widehat{d}(wi)$ it holds that $d(wi)\le d(w)+
(a+1)/(|w|+1)^{c+1}$.

Since $\widehat{d}(wi)\le \widehat{d}(w(1-i))$ it holds that
\begin{eqnarray*}d(wi)&\le& \widehat{d}(wi)+ (a+1)/(|w|+1)^{c+1}\\ &\le& \widehat{d}(w(1-i))+
(a+1)/(|w|+1)^{c+1} \le {d}(w(1-i))+
2(a+1)/(|w|+1)^{c+1}.\end{eqnarray*}

Since $d$ is a martingale, it follows that  $d(wi)\le d(w) +
(a+1)/(|w|+1)^{c+1}$.

Therefore
\[d(x\upto n) \le d(x\upto (n-1)) +(a+1)/n^{c+1}\]
and $x\not\in\regSS[d]$.

\end{proof}

We now have all the ingredients for our main result.

\begin{theorem}
There is an explicit algorithm that computes the binary expansion of
an absolutely normal number $z$ in  online
 nearly linear time.
\end{theorem}

\begin{proof}
The algorithm arises from a  combination of
 Theorem \ref{theo32}, Lemma \ref{lem2x}, Theorem \ref{theo42}, and Theorem \ref{theo51},  notice that all of them
give fully explicit constructions.

As explained in section 2, a real number $z$ is absolutely normal if
none of the martingales  $\dlzp$  succeed exponentially on the
base-$b$ expansion of $z$.

For each $b$, let  $d$   be the polynomially bounded and
 online
 nearly linear
time computable  supermartingale   with a savings account $g'$
defined in Theorem \ref{theo32}\ and construction \ref{cons31}\ as a
conservative substitute of
 $\dlzp$ (that is, $\expSS[\dlz]=\expSS[\dlzp]\subseteq
\regSS[g']$).

Let $d'$ be the online
 nearly linear
time computable and polinomially bounded martingale with $d'(w)\ge
d(w)$ for all $w$ given by Lemma \ref{lem2x}. Notice that since
$d'(w)\ge d(w)$, $g'$ is a savings account for $d'$.

For $b\ne 2$, we  now use Theorem \ref{theo42}\ for $d', g'$ and we
have
 an online
 nearly
linear time computable martingale $d^{(2)}:\strings \to [0, \infty)$
that succeeds on the base-2 expansion of the irrational reals with
base-$b$ expansion in $\regSS[g']$ ($real_b(\regSS[g'])-\Q\subseteq
real_2(\regSS[d^{(2)}])$).

For $b=2$ we directly take $d^{(2)}(w)=d'(w)/d'(\lambda)$. Notice
that $\Q\subseteq real_b(\regSS[d'])$ because
$\seq_b(\Q)\subseteq\expSS[\dlz]$.

For  each $b$ the computation of $d_b=d^{(2)}$ fulfills the
conditions of Theorem \ref{theo51}, so we can compute in
 online
 nearly
linear time a binary sequence $x$ such that, for every $b$,
$x\not\in\regSS[d_b]$ and therefore $real_2(x)\not\in
real_b(\expSS[\dlz])$. So $x$ is the binary expansion of an
absolutely normal number.

\end{proof}

\section{Open Problem}\label{sec6}

Many questions arise naturally from this work, but the following
problem appears to be especially likely to demand new and useful
methods.

As we have seen, normal  numbers are closely connected to the theory
of finite automata.  Schnorr and Stimm \cite{SchSti72}\ proved that
normality is exactly the finite-state case of randomness.  That is,
a real number $\alpha$ is normal in a base $b\ge 2$ if and only if
no finite-state automaton can make unbounded money betting on the
successive digits of the base-$b$ expansion of $\alpha$ with fair
payoffs.  The theory of finite-state dimension \cite{FSD}, which
constrains Hausdorff dimension \cite{Haus19}\ to finite-state
automata, assigns each real number $\alpha$ a finite-state dimension
$\dimFS^{(b)}(\alpha) \in [0,1]$ in each base $b$.  A real number
$\alpha$ then turns out to be normal in base $b$ if and only if
$\dimFS^{(b)}(\alpha)=1$ \cite{BoHiVi05}. Do there exist {\sl
absolutely dimensioned numbers}, i.e., real numbers $\alpha$ for
which $\dimFS(\alpha) = \dimFS^{(b)}(\alpha)$ does not depend on
$b$, and $0<\dimFS(\alpha) <1$?

\section*{Acknowledgments}
The first author's research was supported in part by National
Science Foundation Grants 0652569, 1143830, 1247051, 1545028, and
1900716. Part of this author's work was done during a sabbatical at
Caltech and the Isaac Newton Institute for Mathematical Sciences at
the University of Cambridge,  part was done during three weeks at
Heidelberg University, with support from the Mathematics Center
Heidelberg and the Heidelberg University Institute of Computer
Science, and part was done during the workshop ``Normal Numbers:
Arithmetical, Computational and Probabilistic Aspects'' at the Erwin
Schr\"odinger International Institute for Mathematics and Physics at
the University of Vienna.

The second author's research was supported in part by Spanish
Government
   MEC Grants  TIN2011-27479-C04-01, TIN2016-80347-R, and PID2019-104358RB-I00.  Part of this author's work was
  done during a research stay at the Isaac Newton Institute for Mathematical Sciences at the University of Cambridge,
  part was done during three weeks at Heidelberg University, with support from the Mathematics Center Heidelberg and the
  Heidelberg University Institute of Computer Science, and part was done during the workshop ``Normal Numbers:
Arithmetical, Computational and Probabilistic Aspects'' at the Erwin
Schr\"odinger International Institute for Mathematics and Physics at
the University of Vienna.

We also thank two anonymous reviewers of this paper for very useful
and detailed comments.

\bibliographystyle{abbrv}
\bibliography{../../biblios/todo}

\end{document}